\documentclass{article}
\usepackage[utf8]{inputenc}

\usepackage[numbers,sort]{natbib} 
\usepackage{amsmath} 
\usepackage{amsthm} 
\usepackage[table,xcdraw]{xcolor} 
\usepackage{graphicx} 
\usepackage{tikz} 

\newtheorem{myTheorem}{Theorem}

\newtheorem{myDefinition}{Definition}

\title{Securities Based Decision Markets}
\author{Wenlong Wang and Thomas Pfeiffer}

\begin{document}

\maketitle

\begin{abstract}
 Decision markets are mechanisms for selecting one among a set of actions based on forecasts about their consequences. Decision markets that are based on scoring rules have been proven to offer incentive compatibility analogous to properly incentivised prediction markets. However, in contrast to prediction markets, it is unclear how to implement decision markets such that forecasting is done through the trading of securities. We here propose such a securities based implementation, and show that it offers the same expected payoff as the corresponding scoring rules based decision market. The distribution of realised payoffs, however, might differ. Our analysis expands the knowledge on forecasting based decision making and provides novel insights for intuitive and easy-to-use decision market implementations.
\end{abstract}


\section{Introduction}
Prediction markets \cite{Plott2000MarketsTools, Plott2003ParimutuelResults,Hanson2003CombinatorialDesign,Manski2006InterpretingMarkets,Wolfers2006InterpretingProbabilities,Tziralis2007PredictionReview,Hanson2007LogarithmicAggregation,Berg2008PredictionRun,Arrow2008TheMarkets,Chen2010GamingMaker,Dreber2015UsingResearch} are popular tools for aggregating distributed information into often highly accurate forecasts. Participants in prediction markets trade contracts with payoffs tied to the outcome of future events. The pricing of these contracts reflects aggregated information about the probabilities associated with the possible outcomes. A frequently used contract type is Arrow-Debreu securities that pay \$1 when a particular outcome is realised, and otherwise pay \$0. If such a security is traded at \$0.30, this can be interpreted as forecast for that outcome to occur at 30\% chance. Potential caveats with the interpretation of prices in prediction markets as probabilities have been discussed in the literature \cite{Manski2006InterpretingMarkets,Wolfers2006InterpretingProbabilities}, but are not seen as critical for typical applications \cite{Hanson2003CombinatorialDesign,Manski2006InterpretingMarkets,Wolfers2006InterpretingProbabilities}. Prediction markets have been extensively investigated in lab based experiments and real world settings \cite{Plott2000MarketsTools,Plott2002InformationProblem,Plott2003ParimutuelResults,Wolfers2006InterpretingProbabilities,Manski2006InterpretingMarkets,Tziralis2007PredictionReview,Berg2008PredictionRun,Dreber2015UsingResearch}.

In many practical prediction markets applications, such as recreational markets on political events, participants trade directly with each other, and one participant’s gain is the other participant’s loss. Prediction markets can, however, also be designed to offer net benefits to the participants. Such incentivised prediction markets can be used by a market creator who is willing to compensate the market participants for the information obtained from the market \cite{Hanson2003CombinatorialDesign,Hanson2007LogarithmicAggregation,Chen2010GamingMaker}. Incentivised prediction markets rely on market maker algorithms to trade with the participants, and on cost functions to update prices based on past transactions. These functions are closely related to proper scoring rules such as the Brier (or quadratic) scoring rule and the logarithmic scoring rule \cite{Bickel2007SomeRules,Gneiting2007StrictlyEstimation}, which measure the accuracy of forecasts and allow rewarding a single expert based on forecast and actual outcome. The market maker in an incentivised prediction market subsidises the entire market rather than single experts; its worst-case loss is finite and its expected loss depends on how much the participants ‘improve’ on the information entailed by the initial market maker pricing \cite{Hanson2003CombinatorialDesign}.

Accurate forecasts, as obtained from prediction markets, can be of tremendous value for decision makers. Commercial companies, for instance, can benefit substantially from accurate forecasts regarding the future demand for their products. However, many decision-making problems require conditional forecasts \cite{Chen2014ElicitingMaking}. To decide, for instance, between alternative marketing campaigns, a company needs to understand how each of the alternatives will affect sales. In other words, it needs to predict, and choose between, `alternative futures'. To implement such forecasting based decision making, Hanson \cite{Hanson1999DecisionMarkets} proposed so called decision markets. While it is non-trivial to properly incentivise participants to provide their information in decision markets, it has been shown that this can be achieved \cite{Othman2010DecisionMarkets,Chen2011DecisionIncentives,Chen2011InformationMaking,Boutilier2012ElicitingMakers,Chen2014ElicitingMaking}. 

Properly incentivised decision markets work in a stepwise process to select one among a number of mutually exclusive actions. First, forecasts about the expected future consequences of each action are elicited in a step analogous to incentivised prediction markets. Second, a decision rule is used to select an action based on the forecasted consequences. Once an action has been selected, and its consequences are revealed, payoffs are provided for the forecasts as elicited in the first step. Importantly, the decision rule in properly incentivised decision markets is stochastic, with each action being picked with a strictly positive probability \cite{Chen2011InformationMaking,Chen2014ElicitingMaking}. Payoffs are scaled up to ensure that the participants' expected payoffs in decision markets remain analogous to those made in properly incentivised prediction markets \cite{Chen2011InformationMaking,Chen2014ElicitingMaking}, and that game-theoretical results on strategic interactions between participants in prediction markets \cite{Chen2010GamingMaker} carry over. 

The literature on decision markets has so far focused on implementations based on scoring rules. For prediction markets it is well established how to implement properly incentivised forecasting such that forecasts are made through the trading of securities. A similar securities based decision market implementation has however not yet been described. Such an implementation is important because participants in decision markets are likely familiar with ordinary asset trading, and it is thus convenient for them to report their forecasts through a securities trading interface. Furthermore, securities based decision markets simplify managing liabilities, because the payment for the purchased securities covers the traders' worst-case loss. We here propose such a securities based decision market setting, and compare it to existing, scoring rule based decision markets.

The remaining manuscript is organised as follows: In section \ref{sec:background} we briefly introduce scoring rules, sequentially shared scoring rules, prediction markets and scoring rule based decision markets. In section \ref{sec:strictlySecurity}, we describe a securities based market design and compare it with the existing scoring rule based decision markets. In section \ref{sec:worst-case losses} we compare our design with the the scoring rule based decision markets mechanism in terms of worst-case losses. In section \ref{sec:liability_distribution_framework}, we discuss how trading in this setup allows to re-allocate worst-case losses. Finally, in section \ref{sec:conclusion}, we conclude and discuss our future work.

\section{Related Work and Notation} \label{sec:background}
\subsection{Scoring Rules}
Let us define $\Omega$ as a finite set of mutually exclusive and exhaustive outcomes $\{\omega_1, \omega_2, ..., \omega_n\}$. A probabilistic prediction for those outcomes is denoted by $\vec{r} = (r_1, r_2, ..., r_n )$ with $\sum_{x=1}^n r_x = 1$ and $r_i \in [0, 1]$. A scoring function $s_i(\vec{r})$ allows to quantify the accuracy of prediction $\vec{r}$  once the outcome $\omega_i$ materialises \cite{Garthwaite2005StatisticalDistributions}.

Scoring rules allow to incentivise forecasters for predictions. Denoting the reported distribution as $\vec{r} = (r_1, r_2, \dots, r_n)$ and the forecasters' belief as $\vec{p} = (p_1, p_2, \dots, p_n)$, the expected payoff for a forecaster is given by
\[
    G(\vec{p}, \vec{r}) = \sum_{k=1}^n p_k s_k(\vec{r})
\]
A scoring rule is defined as proper if a forecaster maximises his/her expected payoff by truthfully reporting what he/she believes.
\[
    G(\vec{p}, \vec{p}) \geq G(\vec{p}, \vec{r})
\]
Furthermore, a scoring rule is strictly proper if $G(\vec{p},\vec{p})>G(\vec{p},\vec{r})$ for all $\vec{r} \neq \vec{p}$. 


\subsection{Sequentially Shared Scoring Rules} \label{sec:sequentially shared scoring rules}
Because information is often distributed across multiple agents, it is of interest to expand proper scoring to elicit forecasts from groups of forecasters. In his work on incentivised prediction markets, Hanson proposed a mechanism to sequentially elicit information from forecasters \cite{Hanson2003CombinatorialDesign,Hanson2007LogarithmicAggregation}. The mechanism keeps a current report $\vec{r}$ and offers a contract for a new report $^*\vec{r}$ to be scored as $s_i(^*\vec{r}) - s_i(\vec{r})$ if the outcome $\omega_i$ is observed. Note that $s_i(^*\vec{r}) - s_i(\vec{r})$ is a proper scoring rule if $s_i$ is a proper scoring rule. Once a forecaster accepts the offer, the decision maker will update the current report from $\vec{r}$ to $^*\vec{r}$ and allow a next forecaster to further modify the new current report. Under such a sequentially shared scoring rule, forecasters are scored for how much they improve or worsen the current report. Such a mechanism uses incentives efficiently in that it avoids paying for the same information twice.


\subsection{Securities based Prediction Markets} \label{sec:securities_based_prediction_markets}
The mechanism described in section \ref{sec:sequentially shared scoring rules} involves a two-sided liability. The decision maker is liable to pay each forecaster who improves a forecast, and forecasters are liable to pay the decision maker if they worsen a forecast. It is often considered convenient to implement sequentially shared scoring rules through the trading of Arrow-Debreu securities \cite{Hanson2003CombinatorialDesign,Hanson2007LogarithmicAggregation}. In such an implementation, forecasters purchase securities from the market maker and their payments cover their liabilities. Another reason to use securities based trading is that the majority of existing real-world prediction markets, such as recreational markets on sports or political events, are trading securities in a double auction process. Traders who are familiar with these prediction markets will prefer an interface to be expressed in terms of trading with securities. 

To incentivise traders, securities are bought and sold by a market creator. The market creator uses a market maker algorithm which keeps track of past trades and sets security prices derived from a cost function. The total amount spent on purchasing a particular quantity of securities can be calculated from this cost function. We denote the quantity of outstanding securities as $\vec{q} = (q_1, q_2, \dots, q_n)$ for a market on $n$ mutually exclusive and exhaustive outcomes $\Omega$. Element $q_i$ represents the number of securities sold by the market creator that pay if outcome $\omega_i$ is observed. The instantaneous prices of the securities with outstanding quantities $\vec{q}$ are denoted as $\vec{r}(\vec{q})$ and play the same role as reports in scoring rule based markets.

Assume a trader wants to change the outstanding securities distribution from $\vec{q}$ to $^*\vec{q}$ by buying securities to change the price from $\vec{r}$ to $^*\vec{r}$. The cost for the trader to purchase the amount of securities $^*\vec{q} - \vec{q}$ can be calculated from $C(^*\vec{q}) - C(\vec{q})$, where $C(\vec{q})$ denotes a cost function. Once the final event, i.e. $\omega_i$ is observed, the market maker will resolve the market by paying \$1 for each winning security. If the trader holds  $^*q_i - q_i$ securities when the market is resolved,  his/her payout will be \$$(^*q_i - q_i)$. Overall the realised payoff for the trader will be
\[
    (^*q_i - q_i) - (C(^*\vec{q}) - C(\vec{q}))
\]

\citeauthor{Chen2007AMakers} generalised the relationship between cost functions, price functions and scoring rules, and proposed three equations that establish their equivalence \cite{Chen2007AMakers}:

\begin{equation} \label{eq:cost_function}
    \begin{cases}
               s_i(\vec{r}) = q_i - C(\vec{q}) \qquad \forall i\\
               \sum_i r_i(\vec{q}) = 1\\
               r_i(\vec{q}) = \frac{\partial C(\vec{q})}{\partial q_i}
    \end{cases}
\end{equation}
Furthermore, \citeauthor{Chen2010ALearning} proved that there exist a one-to-one mapping between any strictly proper scoring rule and cost function in securities based prediction markets and such a securities based market is incentive compatible \cite{Chen2010ALearning}. Elicitation through scoring rule based and securities based prediction markets offers the same payoffs for participants when the markets start with the same initial forecasts and end with the same final forecasts.

Cost functions $C(\vec{q})$ and price functions $r_k(\vec{q})$ have the following properties:

\begin{equation} \label{eq:cost_property}
    \begin{split}
        &C(\vec{q} + \beta\vec{1}) = \beta + C(\vec{q}) \\
        &r_k(\vec{q} + \beta\vec{1}) = r_k(\vec{q}) \\
    \end{split}
\end{equation}
where $\beta$ is a real constant \cite{Chen2007AMakers}. These properties imply that the same report can be made through different trades. If a trade $^*\vec{q} - \vec{q}$ changes market prices from $\vec{r}$ to $^*\vec{r}$, so does a trade $^*\vec{q}+\beta \vec{1} - \vec{q}$. This permits the trader to make any report by buying contracts from the market creator. Short selling is not required. The overall payoff will not be affected by the choice of $\beta$, because both costs of purchasing the contracts, and the payout from the contracts at resolution increase by the same amount.






\subsection{Decision Markets} \label{sec:decisionMarkets}

The design of decision markets expands prediction markets to use conditional forecasts for decision making. Decision markets consist of two components. The first component is a set of conditional prediction markets, each of which elicits the forecasts for one of the actions. The second component is the decision rule that defines--- based on conditional prediction markets forecasts--- how the final decision will be made. An example is the MAX decision rule \cite{Othman2010DecisionMarkets} which is to always select the action that has the highest predicted probability for a desired outcome to occur. 

Decision markets with deterministic rules such as the MAX decision rule do not always properly incentivise a forecaster to truthfully report irrespective of the scoring rule it uses \cite{Othman2010DecisionMarkets,Chen2011InformationMaking}. An intuitive example to illustrate how a trader can benefit from misreporting is given in \cite{Othman2010DecisionMarkets}. \citeauthor{Chen2014ElicitingMaking} described that a stochastic decision rule can myopically incentivise forecasters to truthfully report \cite{Chen2014ElicitingMaking}. This approach is rephrased in the following in the notation used throughout this paper to allow for straight forward comparison with the securities-based implementation.

\begin{myDefinition} \label{def:decision_markets_general}
In a decision market, the market creator has a finite set of $m$ actions $\mathcal{A}=\{\alpha_1, \alpha_2, ..., \alpha_m\}$ to choose from. For each action $\alpha_j$, there is a set of possible outcomes $\Omega_j = \{\omega^j_{1}, \omega^j_{2}, \dots, \omega^j_{n_j}\}$, which $n_j$ is the number of possible outcomes for action $\alpha_j$. Both action set $\mathcal{A}$ and outcome sets $\Omega_j$ are collectively exhaustive and mutually exclusive. A stochastic decision rule $\vec{\phi}$ assigns a probability $\phi_k$ to each action $\alpha_k$ with $\phi_k > 0$ and $\sum_{k=1}^m \phi_k = 1$.
\end{myDefinition}
Note that in our notation the outcome $\omega^j_i$ for action $\alpha_j$ can be unrelated to $\omega^k_i$ for action $\alpha_k$. In other words, outcomes can be specific to the actions. The sets for the outcome of two actions can be completely disjoint. The decision rule can take the final report into account, i.e. $\vec{\phi}=\vec{\phi}(\vec{r}_1, \vec{r}_2, \dots, \vec{r}_m)$ where $\vec{r}_1, \vec{r}_2, \dots, \vec{r}_m$ are the final reports over the $m$ different actions. It can, for instance, approximate the MAX decision rule by assigning high probabilities to actions with desirable outcomes.

Similar to scoring rules in prediction markets, decision scoring rules can be defined to map forecasts, decisions and outcomes to a real number. For simplicity, we will denote this score as $S^j_i(\vec{r}_j)$ that the selected action is $\alpha_j$ and the observed event is $\omega^j_i$. Assume $s^j_i(\vec{r}_j)$ is a strictly proper scoring rule for conditional market $j$. A decision score for changing the current report from $\vec{r}_j$ to $^*\vec{r}_j$ is given by:
\begin{equation} \label{eq:oneTimePayOffSco}
   S^j_i(^*\vec{r}_j) - S^j_i(\vec{r}_j) = \frac{1}{\phi_j} \left(s^j_i(^*\vec{r}_j) -  s^j_i(\vec{r}_j) \right)
\end{equation}
The expected payoff $G$ of a forecaster in a scoring rule based decision market as defined in definition \ref{def:decision_markets_general} is given by:
\begin{equation} \label{eq:score_expected_payoffs}
    \begin{split}
        G = & \sum_{j=1}^m \phi_j \sum_{i=1}^{n_j} p^j_i \frac{1}{\phi_j} \left(s^j_i(^*\vec{r_j}) - s^j_i(\vec{r}_j) \right) \\
        = &\sum_{j=1}^m \sum_{i=1}^{n_j} p^j_i \left(s^j_i(^*\vec{r_j}) - s^j_i(\vec{r}_j) \right) 
    \end{split}
\end{equation}
where $p^j_i$ denotes the belief of the forecaster which will be identical to $^*\vec{r_j}$ if the scoring rule is strictly proper. The forecaster has the same expected payoff as if he/she participated in $m$ independent and strictly proper prediction markets. Moreover, findings on strategic interaction between traders and incentives for instantaneous revelation of information from \cite{Chen2010GamingMaker} apply as well.

Note that $\phi_j$ in equations \eqref{eq:oneTimePayOffSco} and \eqref{eq:score_expected_payoffs} is the probability for the selected action in the decision rule after the final report. Equation \eqref{eq:score_expected_payoffs} shows that the value of $\phi_j$ does not affect the expected payoff that risk-neutral forecasters seek to maximise. This is important because for scoring rules that depend on the final report, no participant  except for the final forecaster knows the value of $\phi_j$. To provide truthful forecasts forecasters do not need the decision rule $\phi$ and its dependence on the final report as long as they can trust that the rule has full support.

An alternative method is to ask a single expert for a recommendation and use scoring rules to align the incentives of the expert with the market creator's interest \cite{Chen2014ElicitingMaking, Oesterheld2020DecisionRules}.

\subsection{Empirical Work}
There is a substantial body of empirical studies on prediction markets \cite{Plott2000MarketsTools,Plott2002InformationProblem,Plott2003ParimutuelResults,Wolfers2006InterpretingProbabilities,Manski2006InterpretingMarkets,Tziralis2007PredictionReview,Berg2008PredictionRun,Dreber2015UsingResearch}, and a number of empirical studies on decision markets \cite{Teschner2017ManipulationMarkets}. Some of these studies have addressed whether there are any differences in the efficiency of a scoring rule based mechanism vs. a securities based mechanism. \citeauthor{Jian2012AggregationDistribution} conducted laboratory experiments and found the performance of the two mechanisms is similar \cite{Jian2012AggregationDistribution}; but they stated that further validation was required in ‘field settings’, which familiarity with how to report forecasts may be an important factor. While similar experimental comparisons in decision markets do not exist yet, it is worth noting that securities trading is frequently used in the experimental literature about decision markets \cite{Teschner2017ManipulationMarkets}.

\section{Strictly Proper Securities Based Decision Markets} \label{sec:strictlySecurity}

In section \ref{sec:securities_based_prediction_markets}, we discuss the advantages of implementing forecasting through the trading of securities. We here formulate a cost function for securities based decision markets that offers the same expected payoff for participants as a scoring rule based decision market. In a prediction market, the cost function and price function can be calculated by solving the equations \eqref{eq:cost_function} \cite{Chen2007AMakers}. However, because only decision markets with a stochastic decision rule are myopically incentive compatible, the stochastic decision rule needs to be accounted for.

\subsection{Design}
We adopt the cost function approach for prediction markets as described in section \ref{sec:securities_based_prediction_markets}. To account for the stochastic decision rule, the securities traded in this market have payoffs that depend on the selected action.

\begin{myDefinition} \label{def:decision_market_security}
In addition of the notation in definition \ref{def:decision_markets_general}, we denote $\vec{q}_j = (q^j_1, q^j_2, \dots, q^j_{n_j})$ as outstanding securities for the conditional market for action $\alpha_j$. Element $q_i^j$ represents the number of securities sold by the market creator that pay if action $\alpha_j$ is selected and outcome $\omega_i^j$ is observed. The payout per security is denoted by $v_j$, and can depend on the selected action, but is the same for all traders and does not depend on the observed outcome. The payout for all other securities is zero. Cost function, price function and corresponding scoring rule for the conditional market on action $\alpha_j$ and outcome $\omega^j_i$, are denoted by $C_j(\vec{q}_j)$, $r^j_i(\vec{q}_j)$ and $s^j_i(\vec{r}_j)$, respectively, and together fulfil equation 4. 

\end{myDefinition}


\begin{myTheorem}
Let a trader in a securities based decision market as defined in definition \ref{def:decision_market_security} make a trade $^*\vec{q}_j - \vec{q}_j$ to move prices from $\vec{r}_j(\vec{q}_j)$ to $\vec{r}_j(^*\vec{q}_j)$. Then the trader will have the same expected payoff as a forecaster who makes the same forecast in a scoring rule based decision market as described in equation \eqref{eq:score_expected_payoffs} if and only if we set $v_j = 1/\phi_j$ for all action $\alpha_j$.

\end{myTheorem}

\begin{proof}
Let a forecaster in the a scoring rule based decision market change the reports from $\vec{r}_k$ to $^*\vec{r}_k$ for any action $\alpha_k$. The expected payoff of the forecaster is denoted as $G$ and is given in the equation \eqref{eq:score_expected_payoffs}. Let a trader in our securities based decision market change the outstanding securities distribution from $\vec{q}_k$ to $^*\vec{q}_k$ for each action $\alpha_k$ such that prices change from $\vec{r}_k$ to $^*\vec{r}_k$. Then the realised payoff the trader gains from such a trade is given by:
  \[
      v_j \left( ^*q^j_i - q^j_i \right) - \sum_{k=1}^m \Big( C_k(^*\vec{q_k}) - C_k(\vec{q_k}) \Big)
  \]
  where the selected action is $\alpha_j$ and the observed outcome is $\omega^j_i$.
  
  The expected payoffs of the trader is denoted as $\hat{G}$ and we obtain:    


   \begin{equation} \label{eq:expected_payoff_securities}
        \begin{split}
         \hat{G} = &  \sum_{j=1}^m \phi_j \sum_{i=1}^{n_j} p^j_i \left(v_j \left( ^*q^j_i - q^j_i \right) - \sum_{k=1}^m \Big( C_k(^*\vec{q_k}) - C_k(\vec{q_k}) \Big) \right)\\
        = & \sum_{j=1}^m \phi_j v_j \sum_{i=1}^{n_j} p^j_i\left( ^*q^j_i - q^j_i \right) - \sum_{k=1}^m \Big( C_k(^*\vec{q_k}) - C_k(\vec{q_k})  \Big) \\
        \end{split}
  \end{equation}
  Substituting equation \eqref{eq:cost_function} into the equation \eqref{eq:expected_payoff_securities}, we obtain:
  
  \begin{equation} \label{eq:expected_payoff_securities2}
      \begin{split}
        \hat{G} = & \sum^m_{j=1} \sum^{n_j}_{i=1} p^j_i \left( ^*q^j_i - q^j_i \right) - \sum_{k=1}^m \Big( C_k(^*\vec{q_k}) - C_k(\vec{q_k}) \Big) + \sum_{j=1}^m (\phi_j v_j - 1) \sum_{i=1}^{n_j} p^j_i\left( ^*q^j_i - q^j_i \right) \\
        = & G +  \underbrace{\sum_{j=1}^m (\phi_j v_j - 1) \sum_{i=1}^{n_j} p^j_i\left( ^*q^j_i - q^j_i \right)}_{\text{a}}\\
      \end{split}
  \end{equation}
  
  

  
  The expected payoff $\hat{G}$ in a securities based market is equal to the expected payoff $G$ in a scoring rule based market if and only if term a in equation \eqref{eq:expected_payoff_securities2} is zero. One way to achieve this for arbitrary trades is to set the payoffs of the contracts $v_j$ to $1/\phi_j$. Thus $G = \hat{G}$ if $v_j=1/\phi_j$.
  
  An alternative with $v_j \neq 1/\phi_j$ would be to choose $v_j$ such that the vector (dot) product $\vec{a} \cdot \vec{b}$, with vector element $a_j$ being defined as $\sum_{j=1}^m (\phi_j v_j - 1)$ and $b_j$ being defined $\sum_{i=1}^{n_j} p^j_i\left( ^*q^j_i - q^j_i \right)$, becomes zero. This however, would require to make $v_j$ dependent on trade-specific quantities such as the $^*q^j_i$ and contradicts the properties of contract payoffs as defined in definition \ref{def:decision_market_security}.

\end{proof}

\subsection{Distribution of Realised Payoffs} \label{sec:comparison}

Securities based decision markets and corresponding scoring rule based decision markets provide the identical expected payoff for participants under the same conditions. However, the actual distribution of payoffs for the participants are not necessarily the same. In this subsection, we will discuss the difference between securities based decision markets and the corresponding scoring rule based decision market in terms of realised payoffs for participants.

The realised payoffs for a forecaster who changes report $\vec{r_k}$ to $^*\vec{r_k}$ in a scoring rule based decision market is given by equation \eqref{eq:oneTimePayOffSco}. In the securities based market, assume a trader makes a trade to change the price for any action $\alpha_k$ from $\vec{r}_k$ to $^*\vec{r}_k$ . This trade changes the market creator inventory from $\vec{q}_k$ to $^*\vec{q}_k$ and has a cost given by $C_k(^*\vec{q}_k)-C_k(\vec{q}_k)$. Let the market creator select the action $\alpha_j$ to execute and observe the outcome $\omega^j_i$. Using equation \eqref{eq:cost_function} we obtain the realised payoffs for the trader in the securities based decision market:
\begin{equation} \label{eq:security_actual_payoff}
    \begin{split}
         & \frac{1}{\phi_j} \left( ^*q^j_i - q^j_i \right) - \sum_{k=1}^m \Big( C_k(^*\vec{q_k}) - C_k(\vec{q_k}) \Big) \\
        = \quad & \frac{1}{\phi_j} \left(s^j_i(^*\vec{r_j}) - s^j_i(\vec{r_j}) \right) 
         + \underbrace{\frac{1}{\phi_j} \Big(C_j(^*\vec{q_j}) - C_j(\vec{q_j}) \Big) - \sum_{k=1}^m \Big(C_k(^*\vec{q_k}) - C_k(\vec{q_k}) \Big)}_{\text{a}} \\
    \end{split}
\end{equation}
The term a in equation \eqref{eq:security_actual_payoff} shows that there is a difference in realised payoffs of participants between the securities based decision market and the scoring rule based decision market. This difference cancels out when computing the expected payoff of a participant. Although the sign of term a cannot be decided easily, term a will increase when the trader spends more in the selected conditional market and decrease when the trader spends more in the conditional markets that are not selected.

Equation \eqref{eq:security_actual_payoff} shows that the realised payoffs in our securities based decision market can be rewritten as a scoring rule with an additional term a. This term can be interpreted as a `lottery' that costs $I$ and returns $I/\phi$ at probability $\phi$. Equation \eqref{eq:security_actual_payoff} can be generalised as:

\begin{equation} \label{eq:generalised_actual_payoff}
    \frac{1}{\phi_j} \left(s^j_i(^*\vec{r_j}) - s^j_i(\vec{r_j}) \right) + \underbrace{\frac{1}{\phi_j} I_j - \sum_{k=1}^m I_k}_{\text{a}} \\
\end{equation}
Equation \eqref{eq:security_actual_payoff} is a special case of equation \eqref{eq:generalised_actual_payoff} where $I_k$ equals to the total cost the trader spent in conditional market $k$ for each $k$. This can be interpreted as a forecaster making a report in a scoring rule based decision market and simultaneously participates in a `lottery' at a cost that depends on the report. Any securities based decision market can be seen as a market where the market creator offers a scoring rule based decision market along with a `lottery' with an expected payoff of zero. The terms $I_k$ can be chosen such that the payoffs from a scoring rule based decision market are the same as for a securities based decision market. In the following we provide an example to illustrate the differences between the payoffs in securities based and scoring rule based decision markets. In Section \ref{sec:worst-case losses} and Section \ref{sec:liability_distribution_framework} we will show how the additional flexibility in the design of securities based decision markets can be used to reallocate liabilities and worst case losses between traders and market creator.

\subsubsection{Example: investment into a single conditional market.}
To further detail the differences between a securities based decision market and the corresponding scoring rule based decision market, we analyse an example where a trader invests in only one conditional market, the market corresponding to action $\alpha_k$. We focus on a setting where forecasters do not engage in ‘short selling’, and thus only hold positive positions, and the cost $C_k(^*\vec{q}_k)- C_k(\vec{q}_k)$ paid to the market creator is positive. Assume the market creator to select the action $\alpha_j$ and to observe the outcome $\omega^j_i$. We will compare the realised payoffs between securities based decision markets and corresponding scoring rule based decision market under two conditions.

The realised payoffs for a trader who invested into the selected market, i.e. $j=k$, can be found in table \ref{tab:seleceted}. Naturally, the forecaster in corresponding scoring rule based decision market is also assumed to report in the selected conditional market.
\begin{table}[ht]
\centering
\caption{Realised payoff difference between a decision scoring rule based decision market and the corresponding security based decision market when participants invest into or report on the selected action.}
\begin{tabular}{cc}
\hline
\textbf{Market Type} & \textbf{Realised Payoff}  \\ \hline
\rowcolor[HTML]{EFEFEF}
Scoring Rule Based & $\frac{1}{\phi_j} \left(s_i(^*r_j) - s_i(r_j) \right)$ \\
Securities Based & $\frac{1}{\phi_j}(^*q^j_i - q^j_i ) - (C_j(^*\vec{q_j}) - C_j(\vec{q_j}) )$ \\ \hline
\end{tabular}

\label{tab:seleceted}
\end{table}

Using equation \eqref{eq:cost_function} for the realised payoffs of securities based decision market in table \ref{tab:seleceted} we obtain:
\begin{equation} \label{eq:marketDiff}
\begin{split}
    & \frac{1}{\phi_j} \left(^*q^j_i - q^j_i \right) -  \Big(C_j(^*\vec{q_j}) - C_j(\vec{q_j}) \Big) \\
    = \quad & \frac{1}{\phi_j} \left(s_i(^*r_j) - s_i(r_j) \right) + \underbrace{\frac{1-\phi_j}{\phi_j} \Big(C_j(^*\vec{q_j}) - C_j(\vec{q_j}) \Big)}_{\text{positive}}
\end{split}
\end{equation}
Assuming that traders can only hold positive positions, i.e. cannot ``short" securities, equation \eqref{eq:marketDiff} shows that participants gain a larger payoff in securities based decision markets compared to scoring rule based decision markets.


Table \ref{tab:unselected} shows the realised payoffs of the trader and the forecaster investing into, or report, on an unselected conditional market, i.e. $k \neq j$.
\begin{table}[ht]
\centering
\caption{Payoff difference between a decision scoring rule based decision market and the corresponding security based decision market when participants invest into or report on an unselected conditional market.}
\begin{tabular}{cc}
\hline
\textbf{Market Type} & \textbf{Realised Payoff}  \\ \hline
\rowcolor[HTML]{EFEFEF}
Scoring Rule Based & $0$ \\
Securities Based & $- (C_k(^*\vec{q_k}) - C_k(\vec{q_k}))$ \\ \hline
\end{tabular}
\label{tab:unselected}
\end{table}

Changing the prediction for a conditional market corresponding to an action that is not selected has a zero payoff in the scoring rule based decision market, regardless of how accurate the prediction is. This is because in the scoring rule based decision market, unselected conditional markets will be declared void. However, there is a cost for changing a prediction for an unselected market in the securities based market because purchasing shares to changing a prediction is costly.

\section{Worst-case Losses for Participants and Market Creator} \label{sec:worst-case losses}
An analysis of worst-case losses is crucial for practical implementation because it needs to be ensured that all liabilities can be properly resolved. A common way to ensure that all parties can serve their liabilities is through depositing escrows which can cover the worst-case scenario. A further purpose is to understand how liabilities are distributed between market creator and participants.
\subsection{Worst-case Loss for Participants} \label{sec:worstcaseparticipants}

Consider a forecaster in a scoring rule based decision market who reports $^*\vec{r_k}$ when the current prediction is $\vec{r_k}$ for each conditional market $k$. The worst-case loss for this report is given by: 
\begin{equation} \label{eq:p_srb_wl}
    \begin{split}
        & \min_{j, i} \left(S^j_i(^*\vec{r_j}) - S^j_i(\vec{r_j}) \right) \\
        = \quad & \min_{j, i} \frac{1}{\phi_j} \left(s^j_i(^*\vec{r_j}) - s^j_i(\vec{r_j}) \right)
    \end{split}
\end{equation}

From equation \eqref{eq:p_srb_wl} we can tell that the worst-case loss for the forecaster depends on both the decision rule $\phi_j$ and the report $^*\vec{r_j}$ he/she made. The probability $\phi_j$ depends on the decision rule. Small probabilities in the decision rule, which may be in the interest of market creator to approximate deterministic scoring rules, increase the worst-case loss for the forecaster.

A trader in a security based decision market purchase securities from $\vec{q}_k$ to $^*\vec{q_k}$. Assuming again that forecasters cannot hold negative positions, the worst-case loss for the trader can be calculated as:
\begin{equation} \label{eq:p_sb_wl}
    \begin{split}
        & \sum_{k=1}^m \Big(C_k(\vec{q_k}) - C_k(^*\vec{q_k}) \Big) \\
    \end{split}
\end{equation}

Equation \eqref{eq:p_sb_wl} shows that the worst-case loss for a trader in the securities based decision market only depends on the cost the trader spent. In other words, the trader in the securities based market will not be exposed to any liabilities beyond the costs already paid for purchasing the assets. Therefore a securities based implementation has the advantage that it does not need to further track the liabilities on the side of the traders. Moreover, the worst-case loss of trader does not depend on the decision rule.

\subsection{Worst-case Loss for Market Creator} \label{sec:worstcasecreator}
The loss of a market creator mirrors the profits gained by the participants. Apart from the distribution of realised payoffs for the participants, there is therefore a difference of the worst-case loss for the market creators between a scoring rule based decision market and the corresponding securities based decision market.

Carrying over the conditions from equation \eqref{eq:oneTimePayOffSco}, the worst-case loss for a market creator of a scoring rule based decision market is
\begin{equation} \label{eq:mc_srb_wl}
    \begin{split}
        & \min_{j, i} \left(S^j_i( \vec{r_j}) - S^j_i( ^*\vec{r_j}) \right) \\
        = \quad & \min_{j, i} \frac{1}{\phi_j} \left(s^j_i(\vec{r_j}) - s^j_i(^*\vec{r_j}) \right)
    \end{split}
\end{equation}

As we can tell from equation \eqref{eq:mc_srb_wl}, the worst-case loss for a market creator depends on three factors: initial report $\vec{r}_j$ and final report $^*\vec{r}_j$ for the selected conditional market and the decision rule $\phi_j$. Among three factors, the market creator has control over the initial report $\vec{r}_j$ and the value that decision rule $\phi_j$ can take, but does not have control over which action is being picked. Even though the decision rule can be arbitrary as long as forecasters are convinced that it has full support, it is the interest of the market creator to take the final forecasts of the market into account. For instance, it does not fit the interest of the market creator to assign a small probability to the action that is predicted to most likely lead to a desirable outcome. The relationship between decision rule $\phi_j$ and the final score $s^j_i(^*\vec{r_j})$ can be complex in that it depends on how exactly the final forecast determines the decision rule $\vec{\phi}$. There is a suggestion about computing a minimal feasible decision rule for each action according to the budget of market creator \cite{Chen2011DecisionIncentives}.

Using the conditions for equation \eqref{eq:security_actual_payoff}, the worst-case loss for a market creator of a securities based decision market is: 

\begin{equation} \label{eq:secWorstLoss}
    \begin{split}
        & \sum_{k=1}^m \Big(C_k(^*\vec{q_k}) - C_k(\vec{q_k}) \Big) - \max_{i, j} \left(\frac{1}{\phi_j} \left(^*q^{j}_i - q^j_i \right) \right) \\
    \end{split}
\end{equation}

In equation \eqref{eq:secWorstLoss}, the term $\sum_{k=1}^m (C_k(^*\vec{q_k}) - C_k(\vec{q_k})))$ is the income from securities sales for market creator, which mirrors the cost spent by participants in order to move the inventory distribution from $\vec{q_k}$ to $^*\vec{q_k}$ in each conditional market $k$. The second term, $\max_{i, j} (\frac{1}{\phi_j} (^*q^{j}_i - q^j_i) )$ in the equation is the maximal payout that can be won by participants. In order to compare the worst-case losses, we substitute the equation \eqref{eq:cost_function} into the equation \eqref{eq:secWorstLoss} and obtain: 
\begin{equation} \label{eq:secWorstLoss2}
    \begin{split}
        & \sum_{k=1}^m \Big(C_k(^*\vec{q_k}) - C_k(\vec{q_k}) \Big) - \max_{i, j} \left(\frac{1}{\phi_j} \left(^*q^{j}_i - q^j_i\right) \right)\\
        = \quad & \underbrace{\sum_{k=1, k\neq j}^m \Big(C_k(^*\vec{q_k}) - C_k(\vec{q_k}) \Big)}_{\text{a}} - \max_{i, j} \bigg( \underbrace{s^j_i(^*\vec{r_j}) - s^j_i(\vec{r_j})}_{\text{b}}  - \underbrace{\frac{(1-\phi_j)}{\phi_j} \left(^*q^{j}_i - q^j_i \right)}_{\text{c}} \bigg)
    \end{split}
\end{equation}

Term a of equation \eqref{eq:secWorstLoss2} is non-negative, and depends on the sales in all conditional markets except for the one representing the selected action. Term b is the scoring rule that corresponds to our cost function and is bounded. However, term c depends on final outstanding securities $^*q^j_i$ and $(1-\phi_j)/\phi_j$. While the market creator has control over $\phi_j$, the final outstanding securities is not known ex ante. Therefore no finite initial escrow can guarantee to cover the market creator's liability. This loss of a bound on the worst-case loss of a market maker differs from the the loss of a bound from low probabilities in the decision rule as described in \cite{Chen2011DecisionIncentives}. The final budget for a market maker in a decision market does not have an upper limit because traders can buy arbitrarily large numbers of shares $q^j_i$ on the selected action, while buying fewer (or no) shares on the other actions. Note that it is in the interest for the market creator to assign a small probability $\phi_j$ to actions that are not preferred. The term $(1-\phi_j)/\phi_j$ increases rapidly as $\phi_j$ approaches zero. In summary, the advantage of a worst-case loss for the 
participants that does not depend on the decision rule thus comes at the disadvantage that the worst case loss for the market creator cannot be known ex ante.

\subsection{Numeric Example for Worst-case Losses}
Let us consider a numeric example that a market creator has two actions and each action has two outcomes. The market creator uses simple logarithmic scoring rule with $s_j(^*r_j) = log(^*r_j)$ for both actions. The corresponding cost function is $C_j(^*\vec{q}_j) = log{((e^{q_1} + e^{q_2})/ 2})$. The markets start with an initial report of $r^j_1 = r^j_2 = 0.5$ and $q^j_1 = q^j_2 = 0$ for each action $\alpha_j$. A forecaster reports as shown in the table \ref{tab:report}.

\begin{table}[ht]
\centering
\caption{The reports and the corresponding securities.}
\begin{tabular}{c|cc}
 & \multicolumn{2}{c}{\textbf{Report} $\pmb{^*\vec{r}_j}$} \\ 
 & $\omega_1$ & $\omega_2$ \\ \hline
 \rowcolor[HTML]{EFEFEF} 
$\alpha_1$ & 0.88 & 0.12 \\
$\alpha_2$ & 0.27 & 0.73 \\ \hline
\end{tabular}
\qquad
\begin{tabular}{cc}
\multicolumn{2}{c}{\textbf{Securities} $\pmb{^*\vec{q}_j}$} \\ 
$\omega_1$ & $\omega_2$ \\ \hline
\rowcolor[HTML]{EFEFEF}
2 & 0 \\
0 & 1 \\ \hline
\end{tabular}
\label{tab:report}
\end{table}

\begin{figure}[ht]
    \centering
    \includegraphics[width=0.5\textwidth]{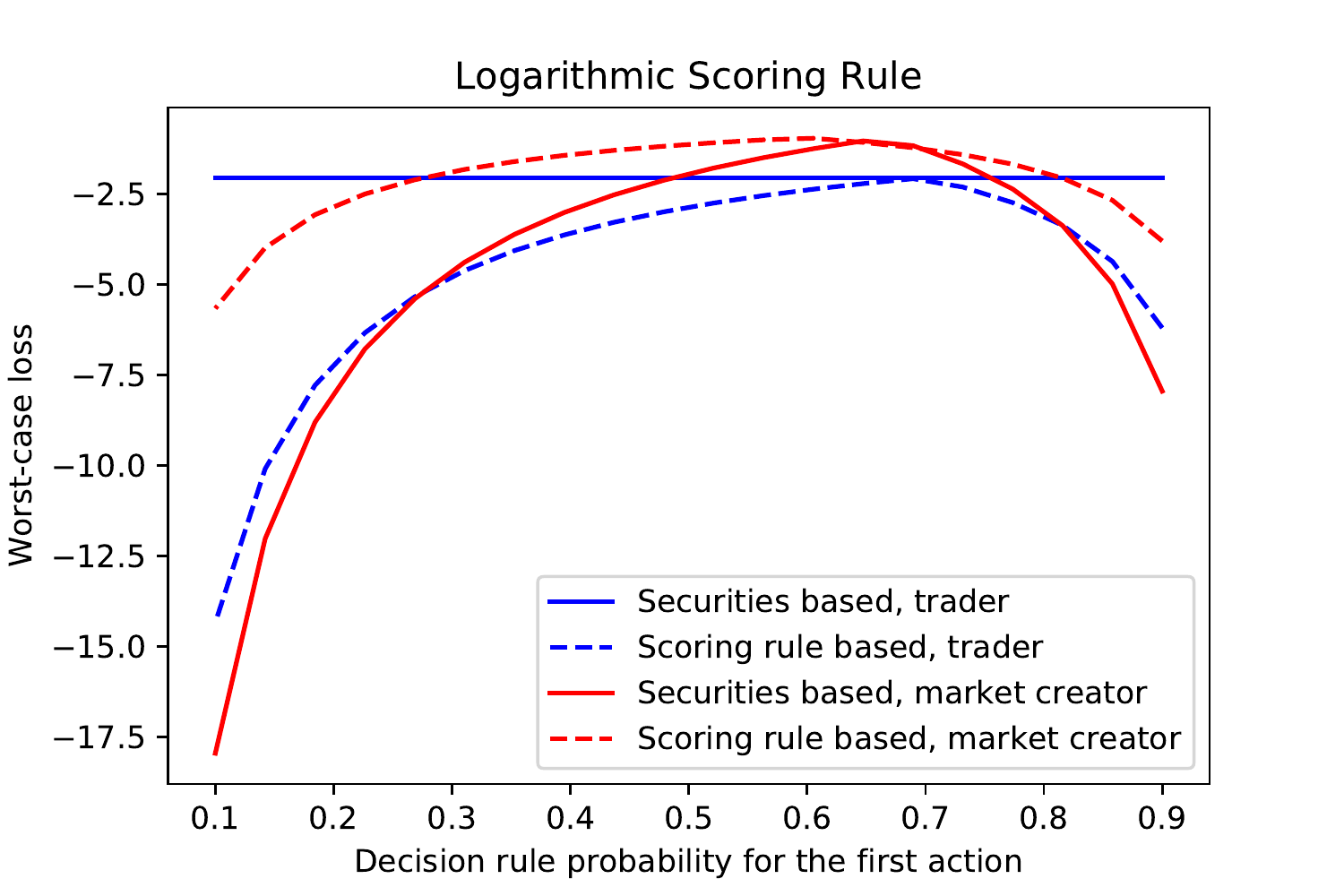}
    \caption{Worst-case loss comparison between the trader and the market creator in different decision markets with a varying probabilities in the decision rule. }
    \label{fig:worst_case_loss_traders_comparison}
\end{figure}

Figure \ref{fig:worst_case_loss_traders_comparison} shows the worst-case loss for the trader and the market creator under a securities based and scoring rule based mechanism in dependence of the decision rule. A trader in the securities based mechanism faces a worst-case loss that does not depend on the decision rule. In the scoring rule based market, his/her risk is higher and increases rapidly if the decision rule approaches 0 or 1. For the values used here, the market creator in the scoring rule based decision market faces less risk than the trader but the risk is still dependent on the decision rule probability. In securities based market, the market creator, however, faces larger risks. The reduction of the worst-case loss of the forecaster comes with an increased worst-case loss for the market creator. This feature of securities based decision markets is preferred by market creators who are willing to sacrifice financial efficiency for more liquidity. 

\section{Re-allocation of worst case losses between market creator and participants} \label{sec:liability_distribution_framework}
Compared to scoring rule based decision markets, securities based decision markets offer additional flexibility to shape the distribution of realized payoffs. This flexibility arises because in a securities based decision market, each report $^*r$ can be realized through infinitely many trades. As outlined in equation \eqref{eq:cost_property}, in prediction markets, if a trader purchases $\beta$ security for each outcome, the cost will be $\beta$. The prices, i.e. the current forecast for the probability distribution over the outcomes, remains the same. The payout from these additional securities will be exact $\beta$ regardless of outcomes, and the net realised payoff for such a trade will be zero.

In a securities based decision market this is, however, not the case. Assume a trader in a securities based decision market purchases a number of $\beta_k $ of each outcome in conditional market $k$, for any $k$. We refer to such a trade as purchasing a ``bundle" of securities. Let the market creator select action $\alpha_j$. Regardless of which outcome is observed, the realised payoff of the trader from these trades can be obtained as:
\begin{equation}
    \frac{1}{\phi_j} \beta_j - \sum^m_{k=1} \beta_k
\end{equation}
While the prices for each outcome in all conditional markets remains the same, the realised payoff for the trader in a decision market is affected by these trades and depends on which action is selected. 



This property allows a trader to adjust the distribution of realised payoffs through purchasing bundles of securities without changing the reported probability distributions. Purchasing the same number of $\beta_k $ of each outcome in conditional market $k$ can also be viewed as the trader purchasing a `lottery' ticket that costs $\beta_k$ and returns $\beta_k/\phi_k$ at a probability of $\phi_k$. The realised payoff of a trader in a securities based decision market can be rewritten as:
\begin{equation} \label{eq:beta_realised_payoff}
    \frac{1}{\phi_j} (^*q^j_i - q^j_i) - \sum^m_{k=1}  (C_k(^*\Vec{q}_k) - C_k(\Vec{q}_k)) + \frac{1}{\phi_j} \beta_j - \sum^m_{k=1} \beta_k
\end{equation}

For each specific report there is a manifold of trading strategies that link to it each of which leading to a different distribution of payoffs.

\subsection{Standardised trades}
For further discussion, we introduce a standardized trade as reference point in the trading strategy space. We define a standardized trade for changing reports from $\vec{r}_k$ from $^*\vec{r}_k$ as the trade with the smallest possible non-negative elements in the number of purchased shares for any $k$. The standardised trade can be obtained from any arbitrary trade $(^*\vec{q}_k - \Vec{q}_k)$ as:
\[
    (^*\vec{q}_k - \Vec{q}_k) - \min_i (^*q^k_i - q^k_i) \times \Vec{1}_k \qquad \forall k
\]

An example is shown in table \ref{tab:std_securities}. An arbitrary and the standardised trade lead to the identical report. A standardised trade is the least costly way to make a report without shorting in a securities based decision market. As for any trade that does not involve short positions, traders will have no liabilities once they have paid the cost for purchasing securities. On the other hand, the market creator has a greater liability to resolve the outstanding securities after the action is selected and the outcome is realised. 
\begin{table}[ht]
\centering
\caption{Both the arbitrary trade and the standardised trade lead to the same report but result in different realised payoffs.}
\begin{tabular}{c|cc}
 & \multicolumn{2}{c}{\textbf{Arbitrary Trade}} \\
 & $\omega_1$ & $\omega_2$ \\ \hline
 \rowcolor[HTML]{EFEFEF}
$\alpha_1$ & 3 & 1 \\
$\alpha_2$ & 1 & 2 \\ \hline
\end{tabular}
\qquad
\begin{tabular}{cc}
\multicolumn{2}{c}{\textbf{Standardised Trade}} \\
$\omega_1$ & $\omega_2$ \\ \hline
\rowcolor[HTML]{EFEFEF}
2 & 0 \\
0 & 1 \\ \hline
\end{tabular}

\label{tab:std_securities}
\end{table}

\subsection{Approximating the payoffs of scoring rule based decision markets}
The flexibility to shape the distribution of realised payoffs can be used to design trades such that the payoffs in securities based decision markets match exactly those in scoring rule based markets. However, this requires the traders to accept negative positions, i.e. to short sell securities, and re-introduces two-sided liabilities.

Let $\beta_k$ in equation \eqref{eq:beta_realised_payoff} to be substituted by $- \big(C_k(^*\Vec{q}_k) - C_k(\Vec{q}_k)\big)$ for all $k$. The realised payoff for such a trader can be obtained by:
\begin{equation} \label{eq:scoring_rule_based_approx}
\begin{split}
    & \frac{1}{\phi_j} (^*q^j_i - q^j_i) - \sum^m_{k=1}  \big(C_k(^*\Vec{q}_k) - C_k(\Vec{q}_k)\big)  - \frac{1}{\phi_j}\big(C_j(^*\Vec{q}_j) - C_j(\Vec{q}_j)\big) \\ & + \sum^m_{k=1} \big(C_k(^*\Vec{q}_k) - C_k(\Vec{q}_k)\big) \\
    = \quad & \frac{1}{\phi_j}\Big((^*q^j_i - q^j_i) - \big(C_j(^*\Vec{q}_j) - C_j(\Vec{q}_j)\big)\Big) \\
    = \quad & \frac{1}{\phi_j}\big( ^*s^j_i(\Vec{r}_j) - s^j_i(\Vec{r}_j)\big)
\end{split}
\end{equation}


With such a trade, a trader makes a report through longing securities that she/he believes are under-priced and shorting securities on over-priced outcomes to meet the cost. The net cost for such a trade is zero. An example for the numbers of securities exchanged under such a trading strategy is shown in Table \ref{tab:scoring_rule_approx}.

\begin{table}[ht]
\centering
\caption{Traders can design their trades to achieve the same realised payoffs as the corresponding scoring rule based decision market. However, this requires short selling.}
\begin{tabular}{c|cccc}
  & \multicolumn{2}{c}{\textbf{Standardised Trades}} & \multicolumn{2}{c}{\textbf{Scoring Rule}}                 \\
  & $\omega_1$         & $\omega_2$         & $\omega_1$                 & $\omega_2$              \\ \hline
  \rowcolor[HTML]{EFEFEF}
$\alpha_1$ & 2               & 0                 & $2 - \ln{(e^2 + 1)/ 2}$ & $ - \ln{(e^2 + 1)/ 2}$ \\
$\alpha_2$ & 0               & 1                 & $- \ln{(e + 1)/ 2}$     & $1 - \ln{(e + 1)/ 2}$  \\ \hline
\end{tabular}

\label{tab:scoring_rule_approx}
\end{table}

\subsection{Market Creator Liability Free Decision Markets}
We will conclude this discussion a with a design that allocates liabilities entirely to the side of traders. Let the $\beta_k = \max_i (^*q^k_i - q^k_i)$ in the conditional market which represents the action $\alpha_k$, we obtain 

\[
        (^*\vec{q}_k - \Vec{q}_k) - \max_i (^*q^k_i - q^k_i) \times \Vec{1}_k \qquad \forall k
\]

The realised payoff from the market creator point of view can be obtained by:
\begin{equation} \label{eq:beta_max}
    \begin{split}
         & \sum^m_{k=1} \big(C_k(^*\Vec{q}_k) - C_k(\Vec{q}_k)\big) - \sum^m_{k=1} \beta_k - \frac{1}{\phi_j} \left((^*q^j_i - q^j_i) - \beta_j \right) \\
         = \qquad & \sum^m_{k=1} \big(C_k(^*\Vec{q}_k) - C_k(\Vec{q}_k)\big) - \sum^m_{k=1} \max_x (^*q^k_x - q^k_x) - \frac{1}{\phi_j}\left((^*q^j_i - q^j_i) - \max_x (^*q^j_x - q^j_x) \right) \\
         = \qquad & \frac{1}{\phi_j}\underbrace{\left(\max_x (^*q^j_x - q^j_x) - (^*q^j_i - q^j_i) \right)}_{\text{Non-negative}} - \underbrace{\sum^m_{k=1} \left(\max_x \big(s^k_x(^*\vec{r}_k) - s^k_x(\vec{r}_k)\big) \right) }_{\text{Known ex ante}} \\
    \end{split}
\end{equation}

As demonstrated by equation \eqref{eq:beta_max}, the term that depends on decision rules and quantities of outstanding securities is guaranteed to be non-negative. As a result, the liability of the market creator is guaranteed to be covered if the market creator can afford $\sum^m_{k=1} \max_x \big(s^k_x(^*\vec{r}_k) - s^k_x(\vec{r}_k)\big)$ which is known ex ante. 

\begin{table}[ht]
\centering
\caption{In this market, traders report probability distribution through short selling completely.}
\begin{tabular}{c|cccc}
  & \multicolumn{2}{c}{\textbf{Standardised Trades}} & \multicolumn{2}{c}{\textbf{Liability-free}} \\ 
  & $\omega_1$         & $\omega_2$         & $\omega_1$         & $\omega_2$        \\ \hline
  \rowcolor[HTML]{EFEFEF}
$\alpha_1$ & 2               & 0                 & $0$             & $ - 2$           \\
$\alpha_2$ & 0               & 1                 & $-1$            & $0$     \\ \hline
\end{tabular}
\label{tab:liability_free}
\end{table}

In this setting, the market creator essentially purchases securities from the forecasters. The market creator covers the worst case loss by paying for the securities.

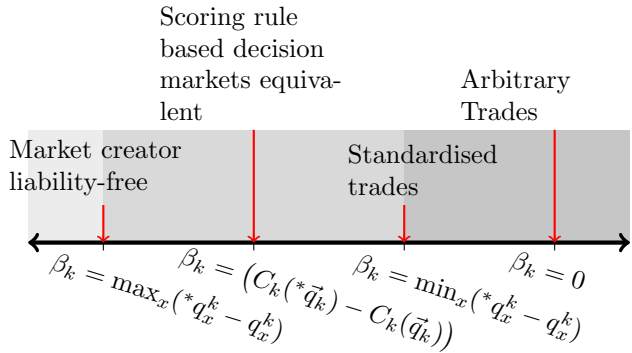
\begin{figure}[ht]
    \centering
    \begin{tikzpicture}
    
    \fill [gray!15] (0,0) rectangle (1, 1.5);
    \fill [gray!30] (1,0) rectangle (5, 1.5);
    \fill [gray!45] (5,0) rectangle (8, 1.5);
    
    \draw[ultra thick, <->] (0,0) -- (8,0);
    
    \foreach \x in {1,3,5,7}
    \draw (\x cm,3pt) -- (\x cm,-3pt);
    
    \draw[ultra thick] (1,0) node[below=3pt,rotate=-17, xshift=1cm] {$\beta_k =  \max_x (^*q^k_x - q^k_x)$};
    \draw[ultra thick] (3,0) node[below=3pt,rotate=-17, xshift=1cm] {$\beta_k =  \big(C_k(^*\Vec{q}_k) - C_k(\Vec{q}_k)\big)$};
    \draw[ultra thick] (5,0) node[below=3pt,rotate=-17, xshift=1cm] {$\beta_k =  \min_x (^*q^k_x - q^k_x)$};
    \draw[ultra thick] (7,0) node[below=3pt,rotate=-17, xshift=0cm] {$\beta_k = 0$};

    \draw[thick, ->, red] (1, 0.5) -- (1, 0) node[above=0.5cm, black, text width=2.5cm]{Market creator liability-free};
    
    \draw[thick, ->, red] (3, 1.5) -- (3, 0) node[above=1.5cm, black, text width=2.5cm]{Scoring rule based decision markets equivalent};
    
    \draw[thick, ->, red] (5, 0.5) -- (5, 0) node[above=0.5cm, black, text width=2.5cm, xshift=0.5cm]{Standardised trades};
    
    \draw[thick, ->, red] (7, 1.5) -- (7, 0) node[above=1.5cm, black,text width=2.5cm]{Arbitrary Trades};

    \end{tikzpicture}
    \caption{The liability is shifted with different choice of $\beta_k$.}
    \label{fig:liability_distribution}
\end{figure}

Figure \ref{fig:liability_distribution} provides a graphical summary of the trading strategies discussed above. In the dark gray area on the right side of the figure, traders cover their worst-case losses when paying for the purchased securities. The worst-case loss for the market creator can become arbitrarily high. In contrast, in the left, light gray part of the graph where traders engage solely in short selling, the market creator will cover their worst-case loss when paying for the securities purchased from the forecasters, and the traders will have liabilities to the market creator. The scoring rule based decision market sits in between where the market creator and forecasters both have liabilities. 

\subsection{External Insurers}
In previous sections, we use the trading of bundles $\beta_k$ in a conditional market $k$ to better understand how liabilities are allocated in different mechanisms. Consider we separate traders into two kinds: regular traders report in standardised trading (thus change prices) with long position only; special traders trade ``bundles" of securities (which does not change prices) with short position only. These special traders can be considered as insurers. In this case, the liability can be separated from both regular traders and the market creator. Thus the stochastic decision rule can further approximate deterministic decision rules with a predictable worst-case loss for the market creator. 

If an insurer accepts short positions at an amount equal to the costs of all regular traders' outstanding securities $\big(C_k(^*\Vec{q}_k) - C_k(\Vec{q}_k)\big)$ for each conditional market $k$, the realised worst-case loss of the market creator can be obtained by:
\[
\begin{split}
    & \sum^m_{k=1}  \big(C_k(^*\Vec{q}_k) - C_k(\Vec{q}_k)\big) - \frac{1}{\phi_j} (^*q^j_i - q^j_i)  \\& + \frac{1}{\phi_j}\big(C_j(^*\Vec{q}_j) - C_j(\Vec{q}_j)\big) - \sum^m_{k=1} \big(C_k(^*\Vec{q}_k) - C_k(\Vec{q}_k)\big) \\
    = \quad & \frac{1}{\phi_j}\Big(\big(C_j(^*\Vec{q}_j) - C_j(\Vec{q}_j)\big) - (^*q^j_i - q^j_i)\Big) \\
    = \quad & \frac{1}{\phi_j}\big( s^j_i(\Vec{r}_j) - ^*s^j_i(\Vec{r}_j)\big)
\end{split}
\]
The realised payoff will be identical to the scoring rule that the cost function is derived from. If we change the insurers' shorting position to $\max_{x}\big(^*q^k_x - q^k_x\big)$ to each conditional market $k$ and substitute equation \eqref{eq:cost_function} into the market creator's realised payoff , then we will have:

\[
    \begin{split}
        & \sum^m_{k=1}  \big(C_k(^*\Vec{q}_k) - C_k(\Vec{q}_k)\big) - \frac{1}{\phi_j} (^*q^j_i - q^j_i)  + \frac{1}{\phi_j}\max_{x}\big(^*q^k_x - q^k_x\big) - \sum^m_{k=1} \max_{x}\big(^*q^k_x - q^k_x\big) \\
         = \qquad & \sum^m_{k=1} \big(C_k(^*\Vec{q}_k) - C_k(\Vec{q}_k)\big) - \sum^m_{k=1} \max_x (^*q^k_x - q^k_x) - \frac{1}{\phi_j}\left((^*q^j_i - q^j_i) - \max_x (^*q^j_x - q^j_x) \right) \\
         = \qquad & \frac{1}{\phi_j}\underbrace{\left(\max_x (^*q^j_x - q^j_x) - (^*q^j_i - q^j_i) \right)}_{\text{Non-negative}} - \underbrace{\sum^m_{k=1} \left(\max_x \big(s^k_x(^*\vec{r}_k) - s^k_x(\vec{r}_k)\big) \right) }_{\text{Known ex ante}} \\
    \end{split}
\]

The insurer essentially offers a lottery. For real-world applications it might not be realistic to assume that this is done without additional costs. However, paying a small additional fee to such an insurer is might be in the interest of the market creator to reduce the worst-case loss, which in turn will also to better approximate deterministic decision rules.

\section{Conclusion and Discussion} \label{sec:conclusion}
We introduce a setting for securities based decision markets that can be conveniently deployed in practical applications. In such a setting, a trader will report a forecast through trading securities. For the securities that represent the selected action and observed outcome, a trader will receive $1/\phi_j$ payoff per share, where $\phi_j$ is the probability in the decision rule corresponding to the selected action. Other shares pay zero, including those purchased in the unselected conditional markets. We prove that under the same condition, specifically, the same action space and the same outcome space, a securities based decision market in our setting has the same expected payoff for participants as the corresponding scoring rule based decision market. 

We compare a securities based decision market under our setting and the corresponding scoring rule based decision market in terms of realised payoffs for participants. The comparison demonstrates that the difference depends on how much the participants report or trade in the selected conditional market. We notably find that the forecaster in scoring rule based decision market will have no cost for reporting in unselected conditional markets while this is not the case in the securities based decision market. We further show that with an additional `lottery', a forecaster in scoring rule based decision market can have the identical realised payoffs as a trader in the corresponding securities based decision market under the same condition. Similarly, the realised payoffs of a trader in a securities based decision market can recover the realised payoffs of the corresponding scoring rule based market, but this requires the forecasters to ‘short-sell’ securities, i.e. to hold negative positions.

By being equivalent to a scoring rule based decision market with an additional zero-mean `lottery', the securities based mechanism described here offers an additional set of parameters that allow to shape the distribution of payoffs beyond what can be achieved based on scoring rules alone. This allows to re-allocate liabilities and worst-case losses between forecasters and the market creator. As illustrated in section \ref{sec:worst-case losses}, in a market where forecasters only purchase positive positions (no short selling), their liabilities are covered when paying for the purchased securities. Moreover, in contrast to scoring rule based decision markets, their worst-case losses do not depend on the probabilities used in the decision rule. A securities based decision market design might thus be of advantage for a market creator who aims to attract forecasters who are concerned about limiting their worst-case losses. In Section \ref{sec:liability_distribution_framework}, we show that the creator risks can be further mitigated by external insurers, which allows to closer approximate deterministic decision rules. Further empirical studies will be of value in determining how to shape trading to obtain the most accurate forecasts.

%
%
%
%
%


\bibliographystyle{unsrtnat}
\bibliography{references}

\end{document}